\documentclass[a4paper,english]{lipics}

\usepackage{amsfonts,amsmath,amsthm,amssymb}
\usepackage{pgf,tikz}
\usetikzlibrary{arrows,automata,shapes}
\usepackage{stmaryrd}
\usepackage{xspace}

\newcommand{\N}{{\mathbb N}}
\newcommand{\Z}{{\mathbb Z}}
\newcommand{\wrt}{w.r.t. }
\newcommand{\ie}{i.e. }
\newcommand{\resp}{\emph{resp.} }
\newcommand{\id}[1][]{\mathrm{id}_{#1}}
\newcommand{\aut}[1]{{\mathcal #1}}
\newcommand{\dual}[1]{{\mathfrak d}({#1})}
\newcommand{\mot}[1]{{\mathbf {#1}}}
\newcommand{\pres}[1]{\langle{#1}\rangle}
\newcommand{\presm}[1]{\pres{{#1}}_{+}}
\newcommand{\mz}{\mathfrak m}
\newcommand{\dz}{\mathfrak d}
\newcommand{\clot}[1]{{\mathfrak c}(\aut{#1})}
\newcommand{\clotx}[1]{\llbracket{#1}\rrbracket}
\newcommand{\portraitA}[2][\infty]{{\mathfrak{p}}_{#1}({#2})}
\newcommand{\portrait}[2][\infty]{{\mathfrak{p}}_{#1}\llbracket{#2}\rrbracket}
\newcommand{\I}[1]{{\cal I}_{#1}}
\newcommand{\J}[2][{\cal J}]{{{#1}}\lfloor{#2}\rfloor}
\newcommand{\toto}{\portrait[k]{s}\lfloor{\tau_1,\tau_2}\rfloor}
\newcommand{\GAP}{{\sf GAP}\xspace}

\theoremstyle{plain}
\newtheorem{proposition}[theorem]{Proposition}
\newtheorem{myremark}[theorem]{Remark}

\title{
The finiteness of a group generated by a 2-letter
invertible-reversible Mealy automaton is decidable}
\titlerunning{The finiteness of a 2-letter invertible-reversible automaton group
  is decidable}
\author{Ines Klimann\footnote{The author is partially supported by ANR Project MealyM ANR-JCJC-12-JS02-012-01.}
}
\authorrunning{Ines Klimann}
\affil{Univ Paris Diderot, Sorbonne Paris Cit\'e, LIAFA,\\
    UMR 7089 CNRS, F-75013 Paris,
    France\\\texttt{klimann@liafa.univ-paris-diderot.fr}}
\Copyright[nd]{Ines Klimann}
\subjclass{F.4.3}
\keywords{Mealy automata, automaton semigroups, decidability of
finiteness, decidability of freeness, Nerode equivalence}

\serieslogo{}
\volumeinfo
  {Billy Editor, Bill Editors}
  {2}
  {Conference title on which this volume is based on}
  {1}
  {1}
  {1}
\EventShortName{}
\DOI{10.4230/LIPIcs.xxx.yyy.p}

\begin{document}

  \setlength\jot{0pt}
  \setlength\abovedisplayskip{7pt}
  \setlength\belowdisplayskip{7pt}

\maketitle

\begin{abstract}
We prove that a semigroup generated by a reversible two-state Mealy
automaton is either finite or free of rank~2. This fact leads to the
decidability of finiteness for groups generated by two-state or
two-letter invertible-reversible Mealy automata and to the
decidability of freeness for semigroups generated by two-state
invertible-reversible Mealy automata.
\end{abstract}

\section{Introduction}
\bigskip\emph{Automaton (semi)groups} --- short for semigroups
generated by Mealy automata and groups generated by invertible Mealy
automata ---
were formally introduced half a century ago (for details,
see~\cite{gns,cain} and references therein).
Over the years, important results have started revealing their full
potential, by 
contributing to important conjectures in group theory, as Milnor
problem (first example of a group of intermediate growth) or Burnside
problem (example of a very simple Mealy automaton generating an infinite
torsion group).

In a way, semigroups can be classified according to their growth
function: at one end stand finite semigroups and at the other one
free semigroups. 
Several sufficient or necessary criteria for finiteness of automaton
semigroups
exist~\cite{AKLMP12,KMP12,cain,mal,min,sst,anto,sidkiconjugacy,sidki}, 
but deciding finiteness of such semigroups is still an open problem. As to
freeness, it has been and it is still a
challenge: only some particular invertible Mealy automata, possibly
parametrized, have been shown to generate free
groups~\cite{svv,GM05,nek,VV07,VV10}; and some Cayley automaton
semigroups have been shown to be free~\cite{sst}.

In this paper, we link both issues for semigroups generated by
reversible two-state Mealy automata: we prove that such semigroups are
either finite or free, in this latter case the states of the
generating Mealy automaton being free generators of the semigroup,
answering a conjecture stated in~\cite{KMP12}.
On the basis of this dichotomy between finite and free semigroups, we
prove that finiteness and freeness of the semigroup are decidable if
the generating reversible two-state Mealy automaton is also invertible.
Decidability of finiteness extends by duality to groups generated by
two-letter invertible-reversible Mealy automata. The problems of
deciding finiteness or freeness of automaton semigroups was raised by
Grigorchuk, Nekrashevych, and Sushchanskii~\cite[Problem~7.2.1(b)]{gns}.

Specializing to two letters or states may seem to be a strong
restriction, but most of 
the significant examples in literature have faced this
restriction: the first example of a finitely generated group of intermediate
growth, the Grigorchuk group~\cite{grigorchuk1,gns}, is generated by a
two-letter Mealy automaton while the very smallest Mealy automaton
with intermediate growth~\cite{brs} has two letters and two states;
the lamplighter group~\cite{gz} is generated by
a two-letter and two-state Mealy automaton; the Aleshin 
automaton~\cite{aleshin,VV07} gives the simplest example of a free
automaton group and has two letters. The article~\cite{clas32} is
entirely devoted to the study of groups generated by 3-state 2-letter
invertible Mealy automata.

This paper is organized as follows.
In Section~\ref{sec-sgp} we define Mealy automata and
automaton (semi)groups. Basic tools to manipulate them are
introduced in Section~\ref{sec-tools}. Section~\ref{sec-main} is
devoted to the dichotomy between free and finite semigroups. The
decidability results are proved in
Section~\ref{sec-decidability}.
The cornerstone of our proofs and
constructions is the very classical Nerode equivalence used to minimize
automata.

\section{(Semi)groups generated by Mealy automata}\label{sec-sgp}
\subsection{Mealy automata}
If one forgets initial and final states, a {\em
(finite, deterministic, and complete) automaton} $\aut{A}$ is a
triple 
\(
\bigl( A,\Sigma,\delta = (\delta_i: A\rightarrow A )_{i\in \Sigma} \bigr)
\),
where the \emph{stateset}~$A$
and the \emph{alphabet}~$\Sigma$ are non-empty finite sets, and
where the $\delta_i$ 
are functions.

\smallskip

A \emph{Mealy automaton} is a quadruple 
\(\bigl( A, \Sigma, \delta = (\delta_i: A\rightarrow A )_{i\in \Sigma},
\rho = (\rho_x: \Sigma\rightarrow \Sigma  )_{x\in A} \bigr)\), 
such that both $(A,\Sigma,\delta)$ and $(\Sigma,A,\rho)$ are
automata. 
In other terms, a Mealy automaton is a letter-to-letter transducer
with the same input and output alphabet. 

The graphical representation of a Mealy automaton is
standard, see Figure~\ref{fig-Aleshin-babyAleshin}.

\begin{figure}[h]
\centering
\vspace*{-.7cm}
\subfloat[][The trivial aut.]{
\begin{tikzpicture}[->,>=latex]
\tikzstyle{every state}=[minimum size=12pt,inner sep=0pt]
\node[state] (x) {\(x\)};
\node (so) [left of=x] {};
\node (ne) [right of=x] {};
\node (ee) [right of=ne] {};
\node (oo) [left of=so] {};
\path (x) edge [loop] node[above]{\(a|a\)} (x);
\end{tikzpicture}
}
\quad\subfloat[][The Aleshin automaton.]{
\begin{tikzpicture}[->,>=latex,node distance=1.8cm]
\tikzstyle{every state}=[minimum size=12pt,inner sep=0pt]
\node[state] (z) {\(z\)};
\node[state] (y) [below right of=z] {\(y\)};
\node[state] (x) [above right of=y] {\(x\)};
\node (so) [left of=y] {};
\node (ne) [right of=y] {};
\path (x) edge node[below,near start]{\(b|a\)} (y)
      (x) edge [bend left] node[above]{\(a|b\)} (z)
      (z) edge [bend left] node[above]{\(a|a\), \(b|b\)} (x)
      (y) edge node[below,near end]{\(b|a\)} (z)
      (y) edge [loop right] node{\(a|b\)} (y);
\end{tikzpicture}}\qquad
\quad\subfloat[][The Baby-Aleshin aut.]{
\begin{tikzpicture}[->,>=latex,node distance=1.8cm]
\tikzstyle{every state}=[minimum size=12pt,inner sep=0pt]
\node[state] (x) {\(x\)};
\node[state] (y) [below right of=x] {\(y\)};
\node[state] (z) [above right of=y] {\(z\)};
\node (so) [left of=y] {};
\node (ne) [right of=y] {};
\path (y) edge node[below,near end]{\(a|a\)} (x)
      (x) edge [bend left] node[above]{\(a|b\), \(b|a\)} (z)
      (z) edge [bend left] node[above]{\(b|b\)} (x)
      (z) edge node[below,near start]{\(a|a\)} (y)
      (y) edge [loop right] node{\(b|b\)} (y);
\end{tikzpicture}}
\vspace*{-.3cm}
\caption{Examples of Mealy automata: the Aleshin automaton generates
  the rank 3 free group~\cite{aleshin,VV07}, the Baby-Aleshin automaton
  generates the free product \(\Z_2^{*3}=\Z_2*\Z_2*\Z_2\)~\cite{nek}.}\label{fig-Aleshin-babyAleshin}
\end{figure}
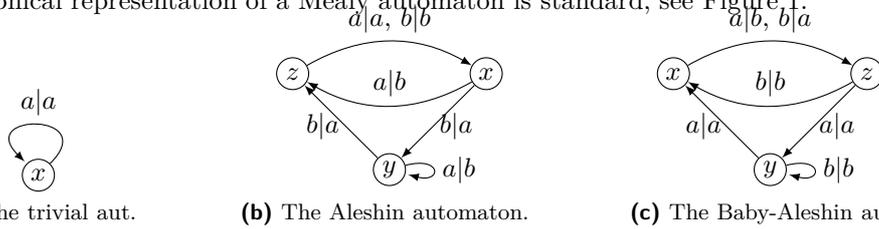

A Mealy automaton \(\aut{A}=(A,\Sigma,\delta, \rho)\) is
\emph{invertible\/} if the functions \(\rho_x\) are permutations of
\(\Sigma\) and \emph{reversible\/} if the functions \(\delta_i\) are
permutations of \(A\).

In a Mealy automaton $\aut{A}=(A,\Sigma, \delta, \rho)$, the sets $A$ and
$\Sigma$ play dual roles. So we may consider the \emph{dual (Mealy)
  automaton} defined by
\(
\dz(\aut{A}) = (\Sigma,A, \rho, \delta)
\).
Obviously, a Mealy automaton is reversible if and only if its dual is
invertible.

Considering the underlying graph of a Mealy automaton, it makes sense
to look at the connected components of a Mealy automaton.
Note that a connected component of a reversible Mealy automaton is
always strongly connected: its \((\delta_i:A\to A)_{i\in\Sigma}\) are
permutations of a finite set and in particular they are surjective.

\subsection{Automaton (semi)groups}\label{sec-semigroups}
Let $\aut{A} = (A,\Sigma, \delta,\rho)$ be a Mealy automaton. 
We view $\aut{A}$ as an automaton with an input and an output tape, thus
defining mappings from input words over~$\Sigma$ to output words
over~$\Sigma$. 
Formally, for $x\in A$, the map
$\rho_x : \Sigma^* \rightarrow \Sigma^*$,
extending~$\rho_x : \Sigma \rightarrow \Sigma$, is defined by:
\begin{equation*}
\forall i \in \Sigma, \ \forall \mot{s} \in \Sigma^*, \qquad
\rho_x(i\mot{s}) = \rho_x(i)\rho_{\delta_i(x)}(\mot{s}) \:.
\end{equation*} 
By convention, the image of the empty word is itself.
The mapping~$\rho_x$ is length-preserving and prefix-preserving.
We say that $\rho_x$ is the \emph{production
function\/} associated with $(\aut{A},x)$ or 
more briefly, if there is no ambiguity,
the \emph{production function\/} of \(x\).
For~$\mot{x}=x_1\cdots x_n \in A^n$ with~$n>0$, set
\(\rho_\mot{x}: \Sigma^* \rightarrow \Sigma^*, \rho_\mot{x} = \rho_{x_n}
\circ \cdots \circ \rho_{x_1} \:\).

Denote dually by $\delta_i:A^*\rightarrow A^*,
i\in \Sigma$, the production functions associated with
the dual automaton
$\dz(\aut{A})$. For~$\mot{s}=s_1\cdots s_n
\in \Sigma^n$ with~$n>0$, set $\delta_\mot{s}: A^* \rightarrow A^*,
\ \delta_\mot{s} = \delta_{s_n}\circ \cdots \circ \delta_{s_1}$. 

\smallskip

The semigroup of mappings from~$\Sigma^*$ to~$\Sigma^*$ generated by
$\rho_x, x\in A$, is called the \emph{semigroup generated
  by~$\aut{A}$} and is denoted by~$\presm{\aut{A}}$.
When $\aut{A}$ is invertible,
its production functions are
permutations on words of the same length and thus we may consider
the group of mappings from~$\Sigma^*$ to~$\Sigma^*$ generated by
$\rho_x, x\in A$; it is called the \emph{group generated
  by~$\aut{A}$} and is denoted by~$\pres{\aut{A}}$.

An invertible Mealy automaton generates a finite group if and only if
it generates a finite semigroup~\cite{AKLMP12}.
A Mealy automaton generates a finite semigroup if and only if so does
its dual~\cite{nek,sv11,AKLMP12}.

\section{Basic tools}\label{sec-tools}
In this section, we present basic tools to manipulate Mealy automata:
Nerode equivalence and minimization of automata 
(\S~\ref{sec-nerode}) are classic constructions from automata
theory, \(\mz\dz\)-reduction and \(\mz\dz\)-triviality
(\S~\ref{sec-mdreduc}) have been introduced in~\cite{AKLMP12} to give
a sufficient condition for finiteness, portraits of automorphisms on
a regular rooted tree (\S~\ref{sec-portraits}) come from geometric
group theory and tensor closures (\S~\ref{sec-synt}) are newly
introduced in order to better control the structure of a Mealy automaton.

\medskip

Let~$\aut{A} = ( A,\Sigma,\delta,\rho)$ be a Mealy automaton.
A convenient and natural operation is to raise \(\aut{A}\)
to the power \(n\), for some \(n>0\): its \emph{\(n\)-th power} is the
Mealy automaton
\begin{equation*}
\aut{A}^n = \bigl( \ A^n,\Sigma, (\delta_i : A^n \rightarrow
A^n)_{i\in \Sigma}, (\rho_{\mot{u}} : \Sigma \rightarrow \Sigma
)_{\mot{u}\in A^n} \ \bigr)\:.
\end{equation*}

Note that the powers of a reversible Mealy automaton are
reversible.

\subsection{Nerode equivalence and minimization of a Mealy automaton}\label{sec-nerode}
{\bf 
Throughout this subsection, $\aut{A}=(A,\Sigma,\delta,\rho)$ denotes a
Mealy automaton.}

The \emph{Nerode equivalence \(\equiv\) on \(A\)} is the limit of the
sequence of increasingly finer equivalences~$(\equiv_k)$ recursively
defined by:
\begin{align*}
\forall x,y\in A,\qquad\qquad x\equiv_0 y & \ \Longleftrightarrow
\ \rho_x=\rho_y\:,\\
\forall k\geqslant 0,\ x\equiv_{k+1} y &
\ \Longleftrightarrow\  \bigl(x\equiv_k y\quad \wedge\quad\forall
i\in\Sigma,\ \delta_i(x)\equiv_k\delta_i(y)\bigr)\:.
\end{align*}
Since the set $A$ is finite, this sequence is ultimately constant;
moreover if two consecutive equivalences are equal, the sequence
remains constant from this point on. The limit is therefore computable.
For every element~$x$ in~$A$, we denote by~$[x]$ (\resp \([x]_k\)) the
class of~$x$ \wrt the Nerode equivalence (\resp the \(\equiv_k\)
equivalence), called 
the \emph{Nerode class\/} (\resp the \emph{\(k\)-class}) of
\(x\). Extending to the \(n\)-th power of \(\aut{A}\), we denote 
by \([\mot{x}]\) the Nerode class in \(A^n\) of
\(\mot{x}\in A^n\).

The \emph{minimization} of $\aut{A}$ is the Mealy automaton
\(\mz(\aut{A})=(A/\mathord{\equiv},\Sigma,\tilde{\delta},\tilde{\rho})\),
where for every $(x,i)$ in $A\times \Sigma$,
$\tilde{\delta}_i([x])=[\delta_i(x)]$ and
$\tilde{\rho}_{[x]}=\rho_x$.
This definition is consistent with the standard minimization of
``deterministic finite automata'' where instead of
considering the mappings $(\rho_x:\Sigma\to\Sigma)_x$, the computation
is initiated by the separation between terminal and non-terminal
states. Using the Hopcroft algorithm, the time complexity of minimization
is \({\cal O}(\Sigma A\log{A})\), see~\cite{ahu74} -- \(E\) being used
here instead of \(\#E\), for a set \(E\), to simplify notations.

Two states of a Mealy automaton belong to the
same Nerode class if and only if they represent
the same element in the generated semigroup, \ie if and only
if they have the same production function \(\Sigma^*\to\Sigma^*\). Two
words on \(A\) of the same length~\(n\) are \emph{equivalent\/} if they
belong to the same Nerode class in \(A^n\). By extension, any two words on
\(A\) are \emph{equivalent} if they have the same
production function.  The set of all words equivalent to \(\mot{x}\in
A^*\), regardless of their length, is denoted by \(\clotx{\mot{x}}\).

Two states of a Mealy automaton belong to the same \(k\)-class if and
only if the restrictions of their production functions to
\(\Sigma^k\to\Sigma^k\) are equal.

The following remarks will be useful for the rest of the paper:
\begin{myremark}\label{rem-lg}
Let \(n\) be an integer. If each word of \(A^n\) is equivalent to a
strictly shorter word, then the semigroup \(\presm{\aut{A}}\) is
finite, its set of elements being \(\{\rho_{\mot{u}}, \mot{u}\in
A^{\leq n-1}\}\).
\end{myremark}

\begin{myremark}\label{rem-rel}
If two words of \(A^*\) are equivalent, so are their images under the
action of each element of~\(\presm{\dual{\aut{A}}}\).
\end{myremark}

\subsection{\(\mz\dz\)-reduction and
  \(\mz\dz\)-triviality}\label{sec-mdreduc}
The \(\mz\dz\)-reduction and the \(\mz\dz\)-triviality were introduced
in~\cite{AKLMP12} to give a sufficient but not necessary condition of
finiteness. We show in Section~\ref{sec-decidability} that, in the case
of a two-state or two-letter invertible-reversible Mealy automaton,
this condition is actually necessary.

A pair of dual Mealy automata is \emph{reduced} if both automata
are minimal. The \emph{$\mz\dz$-reduction} of a Mealy
automaton consists in minimizing the automaton or its dual
until the resulting pair of dual Mealy automata is reduced.
It is well-defined:
if both a Mealy automaton and its dual automaton are non-minimal,
the reduction is confluent~\cite{AKLMP12}.

The trivial Mealy automaton (see
Figure~\ref{fig-Aleshin-babyAleshin}(a)) generates the trivial
(semi)group.
If the \(\mz\dz\)-reduction of a Mealy automaton \(\aut{A}\) leads to
the trivial Mealy automaton, \(\aut{A}\) is said to be
\emph{\(\mz\dz\)-trivial}.
It is decidable whether a Mealy automaton is \(\mz\dz\)-trivial.
An \(\mz\dz\)-trivial Mealy automaton generates a finite
semigroup, but in general the converse is false~\cite{AKLMP12}.

A priori the sequence of mini\-mi\-zation-dualization can be
arbitrarily long: the minimization of a Mealy automaton with a minimal
dual can make the dual automaton non-minimal. Nevertheless, if the
automaton has two states, the \(\mz\dz\)-reduction can be shortened to
\(\mz\dz\mz\dz\). Hence, in this particular case, the time complexity
of the \(\mz\dz\)-reduction is~\({\cal O}(\Sigma\log{\Sigma})\).

\subsection{Portrait of a word}\label{sec-portraits}
{\bf 
Throughout this subsection, \(\aut{A}=(A,\Sigma,\delta,\rho)\) denotes an
invertible Mealy automaton.}

The set \(\Sigma^*\) can naturally
be thought of as a regular rooted tree; its root is the empty word and
two words are connected if and only if they are of the form
\(\mot{s}\) and \(\mot{s}i\), with
\(\mot{s}\in\Sigma^*,\ i\in\Sigma\). The set \(\Sigma^n\) is the
\emph{\(n\)th level\/} of \(\Sigma^*\). A \emph{branch\/} of the tree
\(\Sigma^*\) is a sequence of words \((\mot{s}_k)_{k\in\N}\) such that,
for each \(k\in\N\), \(\mot{s}_k\) is of length \(k\) and is a prefix of
\(\mot{s}_{k+1}\).

An \emph{automorphism\/} of \(\Sigma^*\) is a bijective map
\(\Sigma^*\to \Sigma^*\) preserving the root and the adjacency of
the vertices.
Each state \(x\) of the automaton \(\aut{A}\) acts on the regular
rooted tree \(\Sigma^*\) by the production rule \(\rho_x\).
The constructions of
this subsection are directly inspired by this view  (see~\cite{nek} and
references therein for more details on automorphisms acting on regular
rooted trees). Denote by \(Aut(\Sigma^*)\) the set of automorphisms of
\(\Sigma^*\).

Let \(g\) be an automorphism on the regular rooted tree \(\Sigma^*\).
For any word~\(\mot{s}\in\Sigma^*\), there exists a unique automorphism
\(g_{|\mot{s}}:\Sigma^* \to\Sigma^* \) called a \emph{section\/} of
\(g\) and defined, for all word \(\mot{t}\in \Sigma^*\), by 
\(g(\mot{st})=g(\mot{s})g_{|\mot{s}}(\mot{t})\),
see~\cite{nek} for more details.
The \emph{portrait\/} of \(g\) is the tree \(\Sigma^*\) in which each
vertex \(\mot{s}\in\Sigma^*\) is labeled by
\(g_{|\mot{s}}:\Sigma\to\Sigma\). It is denoted by
\(\portraitA{g}\). The permutation of \(\Sigma\) associated to the
empty word is the \emph{root permutation\/} of \(g\).
 A level (\resp branch) of a portrait is the labeled
level (\resp branch) of the tree.

For a given integer \(k\), the \emph{\(k\)-portrait} of \(g\) is
the  restriction of \(\portraitA{g}\) to levels 0 to~\(k-1\) and
is  denoted by \(\portraitA[k]{g}\), it represents the action of
\(g\) on the partial regular rooted tree \(\Sigma^{\leq k}\).

Let \(\mot{u}\in A^*\). The \emph{portrait\/} (or
\emph{\(\infty\)-portrait} --- \resp the \emph{\(k\)-portrait}) of
\(\mot{u}\) is the portrait (\resp the \(k\)-portrait) of
\(\rho_{\mot{u}}\): each vertex \(\mot{s}\in\Sigma^*\) 
is labeled by \(\rho_{\delta_{\mot{s}}(\mot{u})}:\Sigma\to\Sigma\). It
is denoted by \(\portrait{\mot{u}}\) (\resp
\(\portrait[k]{\mot{u}}\)). This notation is completely justified by
the fact that two equivalent words have the same production
function. An example is given in Figure~\ref{fig-portrait}.

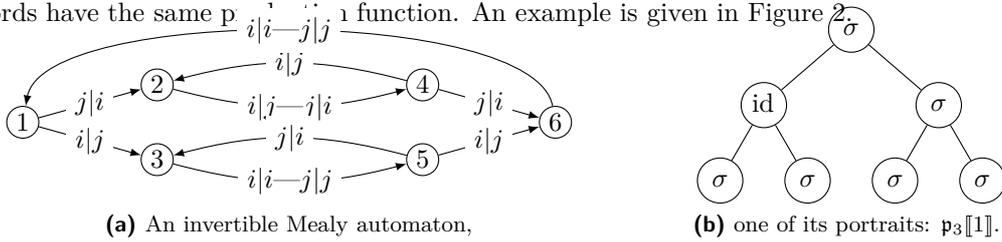
\begin{figure}[h]
\centering
\vspace*{-.7cm}
\subfloat[][An invertible Mealy automaton,]{
\begin{tikzpicture}[->,>=latex,node distance=2.5cm]
\tikzstyle{every state}=[minimum size=12pt,inner sep=0pt]
\node[state] (a) {\(1\)};
\node[state] (b) [below right of=a,above=30] {\(3\)};
\node[state] (c) [above right of=a,below=30] {\(2\)};
\node[state] (h) [right of=b,right=22] {\(5\)};
\node[state] (aa) [right of=c,right=22] {\(4\)};
\node[state] (ab) [right of=a,right=122] {\(6\)};
\path (a) edge node[fill=white]{\(i|j\)} (b)
      (a) edge node[fill=white]{\(j|i\)} (c)
      (b) edge [out=-15,in=195] node[fill=white]{\(i|i\)---\(j|j\)} (h)
      (c) edge [out=-15,in=195] node[fill=white]{\(i|j\)---\(j|i\)} (aa)
      (h) edge node[fill=white]{\(i|j\)} (ab)
      (h) edge [out=165,in=15] node[fill=white]{\(j|i\)} (b)
      (aa) edge [out=165,in=15] node[fill=white]{\(i|j\)} (c)
      (aa) edge node[fill=white]{\(j|i\)} (ab)
      (ab) edge [out=100,in=80,looseness=.5] node[fill=white]{\(i|i\)---\(j|j\)} (a)
;
\end{tikzpicture}
}
\qquad\qquad\subfloat[][one of its portraits: \({\mathfrak
    p}_{3}\llbracket 1\rrbracket\).]{
\begin{tikzpicture}[level/.style={sibling distance=2.3cm/#1,level distance=1cm}]
\tikzstyle{every state}=[minimum size=17pt,inner sep=0pt]
\node[state] {\(\sigma\)}
  child {node[state] {\(\id\)}
           child {node[state] {\(\sigma\)}}
           child {node[state] {\(\sigma\)}}
        }
  child {node[state] {\(\sigma\)}
           child {node[state] {\(\sigma\)}}
           child {node[state] {\(\sigma\)}}
        };
\end{tikzpicture}}
\vspace*{-.3cm}
\caption{Some
 portrait of a two-letter Mealy
  automaton; \(\id=\id[\Sigma]\) and \(\sigma\) permutes \(i\) and
  \(j\).
}\label{fig-portrait}
\end{figure}

\vspace*{-.2cm}

The map from \(Aut(\Sigma^*)\) to the set of portraits induces a
monoid structure on the set of portraits. The neutral element of the
product of portraits is the \emph{identity  portrait}:
\(\I{\infty}=\portraitA{\id[\Sigma^*]}\).
The \emph{portraits of the automaton\/} \(\aut{A}\) are the portraits of
the elements of \(\presm{\aut{A}}\).
The product of two \(k\)-portraits of \(\aut{A}\) can be expressed in
terms of words: \(\portrait[k]{\mot{u}} \portrait[k]{\mot{v}} =
\portrait[k]{\mot{uv}}\). It provides a monoid structure to the set of
\(k\)-portraits of \(\aut{A}\), whose neutral element is the
\emph{identity \(k\)-portrait} \(\I{k}=\portraitA[k]{\id[\Sigma^*]}\).

A level of a portrait is \emph{homogeneous} if all its vertices have
the same label; a portrait is \emph{homogeneous} if all its levels are
homogeneous: the portrait \(\portrait[3]{1}\) of
Figure~\ref{fig-portrait}(b) has homogeneous levels 0 and 2, but is
not homogeneous. For any integer \(k\geq 1\), the \(k\)-portrait
\(\portraitA[k]{g}\) is \emph{almost homogeneous} if 
\(\portraitA[k-1]{g}\) and all the
\(\bigl(\portraitA[k-1]{g_{|i}}\bigr)_{i\in\Sigma}\) are
homogeneous.

An almost homogeneous \((k+1)\)-portrait \({\cal K}\) is built in
the following way from a homogeneous \(k\)-portrait \({\cal J}\)
and a sequence \(\tau=(\tau_i)_{i\in\Sigma}\) of 
permutations of \(\Sigma\): the restriction of \({\cal K}\) to levels
0 to \(k-1\) is \({\cal J}\) and the leaves of the subtree of the root
corresponding to the letter \(i\in\Sigma\) have all label \(\tau_i\). This
portrait is denoted by \(\J{\tau}\), see Figure~\ref{fig-Jtau}.

\begin{figure}[h]
\centering
\vspace*{-.2cm}
\begin{tikzpicture}[level/.style={sibling distance=1.2cm/#1}]
\tikzstyle{every state}=[minimum size=10pt,inner sep=0pt]
\tikzstyle{leaf}=[state]
\tikzstyle{subtree}=[state,isosceles triangle,shape border
  rotate=90,minimum height=7mm]
\node[state] (root) {} [level distance=1cm]
   [child anchor=north]
   child {node[subtree] (s1) {} [level distance=7mm]
          child {node[state]{}}
          child {node[state]{}}
         }
   child {node[subtree] (si) {} [level distance=7mm]
          child {node[state]{\(\tau_i\)}}
          child {node[state]{\(\tau_i\)}}
          node[above=19,right] {\(i\)}
         }
   child {node[subtree] (sn) {} [level distance=7mm]
          child {node[state]{}}
          child {node[state] (dernier) {}}
         };
\node[right of=root,anchor=west] {\({\cal J}\): homogeneous \(k\)-portrait};
\node[left of=si,left=-22] {\ldots};
\node[right of=si,right=-22] {\ldots};
\node[below of=s1,below=-10] {\tiny\ldots};
\node[below of=si,below=-10] {\tiny\ldots};
\node[below of=sn,below=-10] {\tiny\ldots};
\node[above of=root,above=-20] (nord) {};
\node[left of=s1,left=10,below of=s1,below=-24] (sudouest) {};
\node[right of=sn,right=10,below of=sn,below=-24] (sudest) {};
\draw[dashed,] (nord.north) -- (sudouest.south west) -- (sudest.south
east) -- cycle;
\node[right of=dernier,anchor=west] {\(\tau=(\tau_i)_{i\in\Sigma}\)
  sequence of permutations of \(\Sigma\)};
\end{tikzpicture}
\vspace*{-.3cm}
\caption{The almost homogeneous \((k+1)\)-portrait \(\J{\tau}\),
  \(\tau=(\tau_i)_{i\in\Sigma}\).}\label{fig-Jtau}
\end{figure}
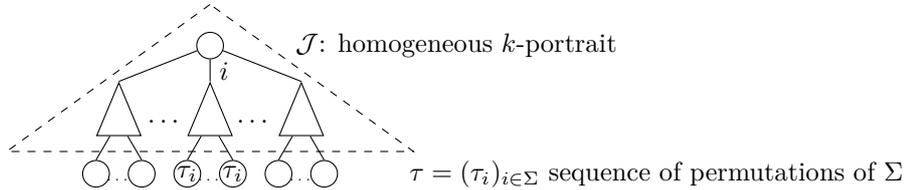

\vspace*{-.2cm}

\begin{myremark}\label{rem-portrait}
The product of two homogeneous \(k\)-portraits is a homogeneous
\(k\)-portrait.

Furthermore, if \(\Sigma=\{i,j\}\):
\begin{itemize}
\item the square of a homogeneous \(k\)-portrait is the identity
\(k\)-portrait \(\I{k}\);
\item the square of an almost homogeneous \(k\)-portrait whose root
  permutation is the identity on \(\Sigma\) is the identity
  \(k\)-portrait;
\item the square of an almost homogeneous \(k\)-portrait
  \(\J{\tau_i,\tau_j}\) whose root permutation is the permutation of
  \(i\) and \(j\) is the identity \(k\)-portrait if and only if
  \(\tau_i=\tau_j\).
\end{itemize}
\end{myremark}

\subsection{Tensor closure}\label{sec-synt}
When a Mealy automaton generates a finite semigroup, we may augment
the alphabet on which it acts to gain a better control over its
structure.

Let \(\aut{A}=(A,\Sigma,\delta,\rho)\) be a Mealy automaton which
generates a finite semigroup. Its \emph{tensor closure\/} is the
Mealy automaton \(\clot{A}=(A,\Xi,\bar\delta,\bar\rho)\), where
\(\Xi=\{\clotx{\mot{s}}\mid \mot{s}\in\Sigma^*\}=\presm{\dual{\aut{A}}}\) and
\(\bar\delta\) and \(\bar\rho\) are the natural extensions of
\(\delta\) and \(\rho\):
\[\forall x\in A,\forall
\mot{s}\in\Sigma^*,\,\bar\delta_{\clotx{\mot{s}}}(x)=\delta_{\mot{s}}(x)\text{ and
}\bar\rho_x(\clotx{\mot{s}}) = \clotx{\rho_x(\mot{s})}\:.\]

A Mealy automaton is \emph{tensor closed\/} if it is
isomorphic to its tensor closure. Its dual is then minimal.

The following remark justifies the introduction of the tensor
closures:

\begin{myremark}\label{rem-clot}
Let \(\aut{A}\) be a two-state Mealy automaton
which generates a finite semigroup. 
Then the automaton \(\clot{A}\) generates a finite semigroup. If
\(\clot{A}\) is \(\mz\dz\)-trivial, then so is \(\aut{A}\).
\end{myremark}
The first result is obtained by looking at the
respective dual automata which generates the same semigroup.
The second result is immediate since a two-state Mealy automaton
\(\aut{A}\) is \(\mz\dz\)-trivial if and only if
\(\mz\dz\mz\dz(\aut{A})\) is trivial and the alphabet of
\(\dz\mz\dz(\aut{A})\) can be injected into the alphabet of
\(\clot{A}\). 

\begin{lemma}\label{lem-complete}
Let \(\aut{A}=(A,\Xi,\delta,\rho)\) be a two-state
invertible-reversible tensor closed Mealy automaton.
The connected components of the powers of \(\aut{A}\) are complete
graphs.
\end{lemma}

\begin{proof}
Let \(k\) be an integer. The connected components of
\(\aut{A}^k\) are strongly connected by
reversibility. Hence any 
two words \(\mot{u}\) and \(\mot{v}\) in the same connected
component are connected by a
path with input label in \(\Xi^*\). The automaton \(\aut{A}\) being
tensor closed,
any word over \(\Xi\) is equivalent to a one-length word over \(\Xi\) and
so the connected component of \(\mot{u}\) and \(\mot{v}\) is a
complete graph: any two states are connected by a transition.
\end{proof}

\section{The semigroup is either free or finite}\label{sec-main}
Recall that a semigroup \(S\) is \emph{free\/} if there exists a subset
\(X\) of \(S\) such that every element of \(S\) can be written
uniquely as a word over \(X\), its \emph{rank\/} is then the
cardinality of~\(X\).
\begin{remark}
On the other hand, a group \(G\) is free if there exists a subset
\(X\) of \(G\) such that every element of \(G\) can be written
uniquely as an irreducible word over \mbox{\(X\sqcup X^{-1}\)}. An
invertible automaton can generate a free semigroup and a non-free 
group; for example, the dual of Aleshin automaton (see
Figure~\ref{fig-Aleshin-babyAleshin}(b)) generates a free semigroup, by
Theorems~\ref{th-main} and~\ref{thm-decidability}, but not a free
group: \(ba^{-1}ba^{-1}=1\).
\end{remark}

\begin{theorem}\label{th-main}
Let \(\aut{A}\) be a reversible two-state Mealy automaton.
If \(\aut{A}\) admits a disconnected power, then it generates a
finite semigroup, otherwise it generates a free semigroup of rank~2
with the states of~\(\aut{A}\) being free generators.
\end{theorem}

Theorem~\ref{th-main} is a corollary of
Proposition~\ref{prop-pas-conn} and the case \(p=2\) in
Proposition~\ref{prop-free} below.

\medskip

Let us look at the connected components of the powers of a Mealy
automaton \(\aut{A}\).
For \(m>0\), \(\mot{u},\mot{v}\in A^m\), and \(x,y\in A\), if there
exists a path from \(\mot{u}x\) to \(\mot{v}y\) in \(\aut{A}^{m+1}\),
then there is a path from \(\mot{u}\) to \(\mot{v}\) in
\(\aut{A}^m\). Hence if \(\aut{A}^n\) is disconnected, so are the 
\(\aut{A}^k\), for all \(k>n\). Thus 
there exists at most one integer \(n\) such
that \(\aut{A}^n\) is connected and \(\aut{A}^{n+1}\) is
disconnected. This integer is called the \emph{connection degree\/} of
\(\aut{A}\). By convention, if \(\aut{A}\) is disconnected, its
connection degree is~\(0\), and it has an infinite connection degree
if no power of \(\aut{A}\) is disconnected. For a Mealy automaton,
having infinite connection degree coincides with the very classical
notion of level transitivity (or spherical transitivity) for its
dual~\cite{nek,gns}.

Note that the Baby Aleshin automaton
(see Figure~\ref{fig-Aleshin-babyAleshin}(c)) is 
reversible, has a connection degree of~2, three states, and generates
an infinite non-free semigroup (its generators have order~2). So
Theorem~\ref{th-main} and Proposition~\ref{prop-pas-conn} 
do not extend to bigger stateset. However, we conjecture that
Proposition~\ref{prop-free} extends to any stateset for invertible
automata.

\subsection{Finite connection degree}
In this section, we prove that a reversible two-state Mealy
automaton has a finite connection degree if and only if it
generates a finite semigroup.
This result is already known~\cite[Lemma~3]{clas32}, but we present
here a new proof; its main idea is to bound the sizes of the
connected components of the powers of \(\aut{A}\) once the connection
degree has passed.

\begin{lemma}\label{lm-conn-inf}
Let \(\aut{A}=(A,\Sigma,\delta,\rho)\) be a reversible Mealy automaton
with at least two states, which generates a semigroup with torsion
elements. Then its connection degree is finite.
\end{lemma}

\begin{proof}
Since \(\presm{\aut{A}}\) has torsion elements, there exist a word
\(\mot{u}\in A^+\) and two integers \(n\geq 0\) and \(k>0\) such that
\(\mot{u}^n\) and \(\mot{u}^{n+k}\) are equivalent:
\(\rho_{\mot{u}^n}= \rho_{\mot{u}^{n+k}}\).

Let \(\mot{s}\in\Sigma^*\), we have: 
\(\delta_{\mot{s}}(\mot{u}^{n+2k}) = \delta_{\mot{s}}(\mot{u}^n)
\delta_{\rho_{\mot{u}^n}(\mot{s})}(\mot{u}^k) \delta_{\rho_{\mot{u}^{n+k}}(\mot{s})}(\mot{u}^k)
= \delta_{\mot{s}}(\mot{u}^n)
\bigl(\delta_{\rho_{\mot{u}^n}(\mot{s})}(\mot{u}^k)\bigr)^2\). 
Hence all the states of the connected component of \(\mot{u}^{n+2k}\)
have form \(\mot{vw}^2\) and \(\aut{A}^{(n+2k)|\mot{u}|}\) is
disconnected.
\end{proof}

{\bf In the 
reminder of this subsection, \(\aut{A}=(A,\Sigma,\delta,\rho)\)
denotes a reversible two-state Mealy automaton (\(A=\{x,y\}\)) with
finite connection degree~\(n\).} If \(z\in A\) is a state of
\(\aut{A}\), \(\bar{z}\in A\) denotes the other state: \(z\neq\bar{z}\).

\begin{lemma}\label{lm-2N}
Let \({\cal C}\) be a connected component of \(\aut{A}^m\) for some
\(m\), and let \(\mot{u}\in A^m\) be a state of \({\cal
  C}\). The connected component (in \(\aut{A}^{m+1}\)) of \(\mot{u}x\)
has size \(\#{\cal C}\) if it does not contain \(\mot{u}y\), and
\(2\#{\cal C}\) if it does contain \(\mot{u}y\).
\end{lemma}

\begin{proof}
Let \({\cal D}\) be the connected component of \(\mot{u}x\):
\(\mot{v}\in A^m\) is a state of \({\cal C}\) if and only if there
exists \(z\in A\) such that \(\mot{v}z\) is a state of \({\cal D}\),
hence: \(N\leq \#\aut{D}\leq 2N\).
Let \(\mot{v}\) be a state of \({\cal C}\) and \(z\in A\):
\(\mot{u}x\) and \(\mot{v}z\) are in the same connected component if
and only if so are \(\mot{u}y\) and \(\mot{v}\bar{z}\). The result
follows.
\end{proof}

Recall that \(n\) is the connection degree of \(\aut{A}\).

\begin{lemma}\label{lm-conn-size}
For each \(m\geq n\), the connected components of \(\aut{A}^m\)
have size exactly \(2^n\).
\end{lemma}

\begin{proof}
By induction on~\(m\geq n\). For \(m\in\{n,n+1\}\), the property is
true (using Lemma~\ref{lm-2N} for \(m=n+1\)).

Assume \(m>n+1\). Suppose that the connected components of
  \(\aut{A}^{m-1}\) and \(\aut{A}^m\) have size~\(2^n\). Then let \({\cal C}\) be a connected
  component of \(\aut{A}^{m+1}\) and \(\mot{u}=u_1\cdots u_{m+1}\) a
  state of \({\cal C}\). The word \(\mot{u}^{\bullet}=u_1\cdots u_m\)
  belongs to a connected component \({\cal D}\) of \({\cal A}^m\), of
  size \(2^n\) by the induction hypothesis. Hence \({\cal C}\) has size
  \(2^n\) or \(2^{n+1}\) according to Lemma~\ref{lm-2N}.

Suppose that \({\cal C}\) has size \(2^{n+1}\): it means by
Lemma~\ref{lm-2N} that both \(\mot{u}\) and
\(\mot{u}^{\bullet}\overline{u_{m+1}}\) belong to \({\cal C}\). It
follows that \(u_2\cdots u_mu_{m+1}\) and \(u_2\cdots
u_m\overline{u_{m+1}}\) belong to the same connected component \({\cal
  E}\) of \(\aut{A}^m\), of size \(2^n\) by the induction
hypothesis. Hence Lemma~\ref{lm-2N} ensures the existence of a
connected component of \({\cal A}^{m-1}\) of size \(2^{n-1}\), 
contradicting the induction hypothesis.
\end{proof}

\begin{proposition}\label{prop-pas-conn}
The connection degree of a reversible two-state Mealy automaton is
finite if and only if it generates a finite semigroup.
\end{proposition}

\begin{proof}
Let \(\aut{A}=(A,\Sigma,\delta,\rho)\) be a reversible two-state Mealy
automaton.
If the connection degree of \(\aut{A}\) is~0,
\(\presm{\dual{\aut{A}}}\) is the trivial semigroup and
\(\presm{\aut{A}}\) is finite~\cite{AKLMP12}.

Otherwise, let \(n\geq 1\) be the connection degree of \(\aut{A}\):
by Lemma~\ref{lm-conn-size}, for \(m\geq n\), the connected components
of \(\aut{A}^m\) have size~\(2^n\). These
connected components are reversible Mealy automata on the alphabet
\(\Sigma\). Up to state numbering, there are only a finite number of
such automata and thus there exist \(p<q\) such that
\(\mz(\aut{A}^p) = \mz(\aut{A}^q)\). It follows by
Remark~\ref{rem-lg} that \(\presm{\aut{A}}\) is finite.

The reciprocal property is a particular case of Lemma~\ref{lm-conn-inf}.
\end{proof}

\subsection{Infinite connection degree}
Here we prove that if a reversible \(p\)-state Mealy automaton, \(p\)
prime, has infinite connection degree, then it generates a
free semigroup, the states of the automaton being free
generators. The idea is to bound the sizes of the
Nerode classes in the powers of \(\aut{A}\).

For the next three lemmas, let \(\aut{A}=(A,\Sigma,\delta,\rho)\) be a
reversible \(p\)-state 
Mealy automaton, \(p\) prime, with infinite connection degree
(\(A=\{x_1,\dots, x_p\}\)).  By Lemma~\ref{lm-conn-inf}, \(\aut{A}\)
generates an infinite semigroup.

\begin{lemma}\label{lm-conn-lg-diff}
There cannot exist
two equivalent words of different length in~\(A^*\).
\end{lemma}

\begin{proof}
For each \(m\), \(\aut{A}^m\) is connected, and so any two words of
length~\(m\) are mapped one onto the other by an element of
\(\presm{\dual{\aut{A}}}\).

Let \(\mot{u}\) and \(\mot{v}\) be two equivalent words of different
lengths, say \(|\mot{u}|<|\mot{v}|\). Every
word of length~\(|\mot{v}|\) is then equivalent to a word of
length~\(|\mot{u}|\): if \(\mot{w}\) is of length \(|\mot{v}|\), then
\(\mot{w} = \delta_{\mot{t}}(\mot{v})\) for some
\(\mot{t}\in\Sigma^*\), and, by Remark~\ref{rem-rel}, \(\mot{w}\) is
equivalent to \(\delta_{\mot{t}}(\mot{u})\) of length \(|\mot{u}|\).
By Remark~\ref{rem-lg}, the semigroup
\(\presm{\aut{A}}\) is finite, which is impossible.
\end{proof}

\begin{lemma}\label{lm-classes2}
All the
Nerode classes of a given power \(A^m\) have the same size, which
happens to be a power of~\(p\).
\end{lemma}

\begin{proof}
Let \(\mot{u}\in A^m\): \([\mot{u}]\subseteq A^m\) by definition. If
\([\mot{u}] = A^m\), the result is clear. Otherwise, let \(\mot{v}\in
A^m - [\mot{u}]\). Since \(\aut{A}^m\) is connected, \(\mot{u}\) is mapped
onto \(\mot{v}\) by an element of \(\presm{\dual{\aut{A}}}\); that is
there exists \(\mot{r}\in\Sigma^*\) such that
\(\mot{v}=\delta_{\mot{r}}(\mot{u})\).

By Remark~\ref{rem-rel}, any word equivalent to \(\mot{u}\) is mapped
by \(\delta_{\mot{r}}\) onto a word equivalent to \(\mot{v}\).
Since the automaton \(\aut{A}^m\) is reversible, \(\delta_{\mot{r}}\)
is a permutation of \(A^m\), hence we find \(\#[\mot{u}] = \#[\mot{v}]\).

The stateset of \(\aut{A}^m\) has size a power of~\(p\), where \(p\) is a
prime number, and so has any Nerode equivalence class.
\end{proof}

\begin{lemma}\label{lm-conn-m-lg}
There cannot exist
two equivalent words of the same length in \(A^*\).
\end{lemma}

\begin{proof}
Let \(\mot{u}\) and \(\mot{v}\) be two different equivalent words of
the same length~\(n+1\). Let us prove by induction on \(m> n\) that
\(\mz(\aut{A}^m)\) has at most \(p^n\) states.

The automaton \(\aut{A}^{n+1}\) has \(p^{n+1}\)
  states. The words \(\mot{u}\) and \(\mot{v}\) are in the same Nerode
  class: by Lemma~\ref{lm-classes2}, all Nerode classes
  of~\({A}^{n+1}\) have at least \(p\)~elements and
  \(\mz(\aut{A}^{n+1})\) has at most \(p^n\) states.

Suppose that  \(\mz(\aut{A}^m)\) has at most \(p^n\) states.
Then, since all Nerode classes have the same size by
  Lemma~\ref{lm-classes2}, the induction hypothesis implies that they
  have at least \(p^{m-n}\) elements. Let us look at \([x_1^m]\): it
  contains
\[x_1^m, \mot{u}_1, \mot{u}_2, \dots, \mot{u}_{p^{m-n}-1}\:,\] which are
pairwise distinct. Among these words, there is at least one whose
suffix in \(x_1\) is the shortest, say \(\mot{u}_1\) without
loss of generality: \(p^{m-n}>1\) and \(x_1^m\) has the longest
possible suffix in \(x_1\). Hence \([x_1^{m+1}]\) contains the following
pairwise distinct \(p^{m-n}+1\) words
\[x_1^{m+1}, \mot{u}_1x_1, \mot{u}_2x_1, \dots,
\mot{u}_{p^{m-n}-1}x_1,x_1\mot{u}_1\:.\] By Lemma~\ref{lm-classes2},
\(\#[x_1^{m+1}]\) is a power of~\(p\), so \(\#[x_1^{m+1}]\geq p^{m+1-n}\). As
all Nerode classes of~\({A}^{m+1}\) have the same cardinality, we can
conclude that \(\mz(\aut{A}^{m+1})\) has at most
\(p^{m+1}/p^{m+1-n}=p^n\) elements, ending the induction.

Consequently, since there is only a finite number of different Mealy
automata with up to \(p^n\) states, there exist \(k<\ell\) such that
\(\mz(\aut{A}^k)\) and \(\mz(\aut{A}^{\ell})\) are equal up to state
numbering. By Remark~\ref{rem-lg}, the semigroup \(\presm{\aut{A}}\)
is finite, which is impossible.
\end{proof}

As a corollary of Lemmas~\ref{lm-conn-inf},~\ref{lm-conn-lg-diff}
and~\ref{lm-conn-m-lg} we can state the following proposition.

\begin{proposition}\label{prop-free}
Let \(\aut{A}\) be a reversible \(p\)-state Mealy automaton, \(p\)
prime. If the automaton~\(\aut{A}\) has infinite connection degree,
then it generates a free semigroup of rank~\(p\) with the
states of \(\aut{A}\) being free generators of the semigroup. The
converse holds for \(p=2\).
\end{proposition}

\section{Decidability of finiteness and of freeness}\label{sec-decidability}
This section is devoted to the decidability of finiteness and of
freeness for semigroups generated by two-state invertible-reversible
Mealy automata by linking Theorem~\ref{th-main} and the possible
\(\mz\dz\)-triviality of such an automaton.

\begin{lemma}\label{lm-identity}
Let \(\aut{A}=(A,\Sigma,\delta,\rho)\) be a two-state
invertible-reversible automaton of finite connection degree
\(n\). Two elements of \(\Sigma^*\) which have the same action on a word
of \(A^n\) are equivalent.
\end{lemma}

\begin{proof}
It is sufficient to prove that \(\id[A^*]\) is the only element of
\(\presm{\dual{\aut{A}}}\) which fixes a word of \(A^n\).

If \(n=0\), \(\presm{\dual{\aut{A}}}\) is the
trivial semigroup and the result is true. Otherwise, let \(\mot{u}\in
A^{n}\) and \(\mot{s}\in\Sigma^*\) such that
\(\mot{u}\) is stable by \(\delta_{\mot{s}}\):
\(\delta_{\mot{s}}(\mot{u}) = \mot{u}\).

By Lemma~\ref{lm-2N}, \(\aut{A}^{n+1}\) has two connected components:
\(\mot{u}x\) belongs to one of them and \(\mot{u}y\) to the other
one. Looking forward, a connected component \({\cal C}\) of
\(\aut{A}^m\), for \(m\geq n\), originates two connected components of
\(\aut{A}^{m+1}\): \(\{\mot{v}z_{\mot{v}}\mid \mot{v}\in{\cal C},
z_{\mot{v}}\in A\}\) and \(\{\mot{v}\overline{z_{\mot{v}}}\mid
\mot{v}\in{\cal C}\}\).
 And all connected
component of \(\aut{A}^{m+1}\) are built this way. Hence if two
different words of the same length~\(m>n\) have the same prefix of
length~\(n\), they belong to different connected components of
\(\aut{A}^m\).

Let \(\mot{t}\in\Sigma^*\) satisfy 
\(\rho_{\mot{u}}(\mot{s}) = \mot{t}\), and let
\(\mot{v}, \mot{w}\in A^*\) such that \(\mot{t}\) maps
\(\mot{v}\) onto \(\mot{w}\): \(\delta_{\mot{t}}(\mot{v}) = \mot{w}\).

The words \(\mot{uv}\) and \(\mot{uw}\) belong to the same connected
component: \[\delta_{\mot{s}}(\mot{uv}) =
\delta_{\mot{s}}(\mot{u})\delta_{\rho_{\mot{u}}(\mot{s})}(\mot{v}) =
\mot{u}\delta_{\mot{t}}(\mot{v}) = \mot{uw}\:,\]
 and have a common prefix of length~\(n\), so
they are equal. Hence: \(\delta_{\mot{t}}=\id[A^*]\). As
\(\dual{\aut{A}}\) is reversible, \(\mot{t}\) is mapped onto
\(\mot{s}\) by an element of \(\presm{\aut{A}}\) and
\(\delta_{\mot{s}}=\text{id}_{A^*}\).
\end{proof}

We have a similar (but weaker) result on shorter words for
tensor closed Mealy automata.
{\bf In the next three lemmas of this section,
  \(\aut{A}=(A,\Xi,\delta,\rho)\) denotes 
a tensor closed two-state invertible-reversible automaton of
finite connection degree~\(n\): \(A=\{x,y\}\).}
By Lemma~\ref{lem-complete}, \(\aut{A}^n\) is complete as a
graph. Furthermore, a transition has a unique label: if a
transition had several labels, they would coincide on a word of \(A^n\) and
by Lemma~\ref{lm-identity} they actually would be the same letter of \(\Xi\).

\begin{lemma}\label{lm-Ak}
Let \(k\) be an integer, \(1\leq k\leq
n\). Two elements of \(\Xi^*\) which map a given word of \(A^k\) into
the same word have the same action on \(A^k\).
\end{lemma}

\begin{proof}
Each word of \(\Xi^*\) is equivalent to a letter of \(\Xi\),
hence it is sufficient to prove the result for letters.

The Mealy automaton \(\aut{A}^n\) has \(2^n\) states, is complete as a
graph and each transition has a unique label, so \(\#\Xi=2^n\). By
hypothesis, \(\Xi\) is the set of elements of
\(\presm{\dual{\aut{A}}}\), so \(\#\presm{\dual{\aut{A}}}=2^n\).

Let us consider the minimization of \(\dual{\aut{A}}\), using the
sequence of increasingly finer equivalences \((\equiv_k)\) introduced
in Section~\ref{sec-nerode}.
Each \(n\)-class of \(\Xi\) is a singleton by Lemma~\ref{lm-identity},
hence the sequence \((\equiv_k)\) remains constant at least from \(n\) on.
So the Nerode equivalence produces \(2^n\) equivalence classes formed
uniquely by singletons, by partitioning the stateset of
\(\dual{\aut{A}}\) of cardinality~\(2^n\) in \(n\) steps, each step
cutting each class of the previous one into at most two subsets as
\(\#A=2\). 
Hence the equivalence
\(\equiv_k\) cuts each \((k-1)\)-class into two sets of the same cardinality:
\(\forall k, 0\leq k\leq n, \forall s\in\Xi,\, \#[s]_k
=\#[s]_{k-1}/2 = 2^{n-k}\).

Let \(k\), \(1\leq k\leq n\), \(\mot{u}\in A^k\), and \(s\in\Xi\). We
have:
\begin{equation}\label{eq-equivk}
[s]_k\subseteq \{t\in\Xi\mid t(\mot{u})=s(\mot{u})\}\:.
\end{equation}
The left set in Equation~\eqref{eq-equivk} has cardinality
\(2^{n-k}\), it is the set of elements of \(\Xi\) which coincide with
\(s\) on \(A^k\). Since two elements of \(\Xi\) whose actions coincide on a
word of \(A^n\) are equivalent, the right set of
Equation~\eqref{eq-equivk} has cardinality at most \(\#A^{n-k} =
2^{n-k}\), and so the two sets of Equation~\eqref{eq-equivk} are
equal, leading to the result.
\end{proof}

One consequence of Lemma~\ref{lm-Ak} is that an element of \(\Xi^*\)
which fixes a word of length \(k\) on \(A\) fixes completely \(A^k\).

Denote by \(\id\) the identity of \(A\) and by \(\sigma\) the
permutation of \(x\) and \(y\).
We can translate Lemma~\ref{lm-Ak} in terms of
portraits of \(\dual{\aut{A}}\): whenever two
\(k\)-portraits of \(\dual{\aut{A}}\)
have an identical branch, they are equal. In particular, \(\I{k}\)
being a portrait of \(\dual{\aut{A}}\), if a whole branch
of a \(k\)-portrait of \(\dual{\aut{A}}\) 
is labeled by \(\id\), this portrait is \(\I{k}\).
Hence if in a \(k\)-portrait of \(\dual{\aut{A}}\), all
vertices at level less than \(k-1\) are labeled by \(\id\), 
this portrait is either \(\I{k}\) or
\(\J[\I{k-1}]{\sigma,\sigma}\). Note that for \(k\leq n\), both
\(\I{k}\) and \(\J[\I{k-1}]{\sigma,\sigma}\) are portraits of
\(\dual{\aut{A}}\).

By Lemma~\ref{lm-identity}, any element of
\(\presm{\dual{\aut{A}}}\) whose \(n\)-portrait is \(\I{n}\) acts
trivially on \(A^*\).

What are the possible portraits of \(\dual{\aut{A}}\)? Since
\(\aut{A}^n\) is connected and \(\aut{A}\) is tensor closed, it is
immediate that each finite sequence \((\pi_i)_{1\leq i\leq n}\in\{\id,\sigma\}^n\) 
labels a branch of an \(n\)-portrait of \(\dual{\aut{A}}\): in
\(\aut{A}^n\), there is a transition with input \(s\in\Xi\) from \(x^n\) to
\(\pi_1(x)\cdots \pi_n(x)\) and the leftmost branch of~\(\portrait[n]{s}\) is
labeled by \(\pi\).

\begin{lemma}\label{lem-portrait}
The portraits of \(\dual{\aut{A}}\) are homogeneous.
\end{lemma}

\begin{proof}
Let us prove the result for \(k\leq n\), by induction on \(k\geq
1\). A \(1\)-portrait has a unique element, its root, and so is
homogeneous.

Suppose that the \(\ell\)-portraits of \(\dual{\aut{A}}\) are all
homogeneous, for \(\ell\leq k<n\). Let us consider a letter
\(s\in\Xi\) and \({\cal S}=\portrait[k+1]{s}\): it is almost
homogeneous by the induction hypothesis. More precisely: \({\cal S} =
\toto\) for \(\tau_1,\tau_2\), some
permutations of \(A\).

\subparagraph*{First case: \(\delta_s\) permutes \(x\) and \(y\).}
We consider the following \((n+1)\)-portrait \({\cal K}\):
\begin{itemize}
\item the restriction of \({\cal K}\) to levels 0 to \((n-k-1)\) is
  \(\I{n-k}\),
\item in bottom-left of \(\I{n-k}\), we put \(\portrait[k+1]{s}\): the
root of \(\portrait[k+1]{s}\) is the left child of the
bottom-left leaf of \(\I{n-k}\)
  (it is possible since we can choose the left branch of a portrait,
  applying Lemma~\ref{lm-Ak} and \(\portrait[k+1]{s}\) is actually a
  portrait of \(\dual{\aut{A}}\)),
\item it is completed to be a portrait of \(\dual{\aut{A}}\).
\end{itemize}

The leftmost branch of \({\cal K}^2\) starts with \(\id^n\). Hence by
Lemma~\ref{lm-identity}, \({\cal K}^2\) is the identity
\((n+1)\)-portrait, which implies \(\tau_1=\tau_2\) by
Remark~\ref{rem-portrait} and Lemma~\ref{lm-identity}, that is \({\cal
  S}\) is homogeneous.

\subparagraph*{Second case: \(\delta_s\) stabilizes \(A\).}
Let \({\cal L}\) be the \((k+1)\)-portrait whose root permutation
is \(\sigma\) and all other vertices are labeled by \(\id\): it is
a portrait of \(\dual{\aut{A}}\) since so are all homogeneous
\((k+1)\)-portraits with root permutation \(\sigma\)
 from first case.
Then by multiplying \({\cal S}\) by \({\cal L}\), we obtain a
non-homogeneous \((k+1)\)-portrait with root permutation \(\sigma\)
which has to be
a portrait of \(\dual{\aut{A}}\). That is impossible.

\medskip

The proof is similar for \(k>n\), considering the portrait
\(\portrait[k]{s}\).
\end{proof}

\begin{lemma}\label{lem-trivial}
The states of \(\aut{A}\) are equivalent.
\end{lemma}

\begin{proof}
By Lemma~\ref{lem-portrait}, all the portraits of \(\dual{\aut{A}}\) are
homogeneous. For any letter \(s\in\Xi\), since its
portrait is homogeneous, \(\rho_x(s)\) and \(\rho_y(s)\) are
equivalent. The automaton being tensor closed, they are equal, and so
\(\rho_x=\rho_y\).
\end{proof}

\begin{theorem}\label{thm-decidability}
Let \(\aut{A}\) be a two-state invertible-reversible Mealy
automaton. It generates a finite group if and only if it is
\(\mz\dz\)-trivial.
\end{theorem}

\begin{proof}
By~\cite{AKLMP12}, if \(\aut{A}\) is \(\mz\dz\)-trivial, it generates
a finite group.

Suppose that \(\aut{A}\) generates a finite group and
consider its tensor closure \(\clot{A}\): \(\clot{A}\) generates a
finite group by Remark~\ref{rem-clot}. The connection degree of
\(\clot{A}\) is finite by Proposition~\ref{prop-pas-conn} and so
\(\clot{A}\) is \(\mz\dz\)-trivial by Lemma~\ref{lem-trivial}. Hence
\(\aut{A}\) is \(\mz\dz\)-trivial by Remark~\ref{rem-clot}.
\end{proof}

The last theorem summarizes all the decidability results arising from
this article.

\begin{theorem}
It is decidable whether a two-state invertible-reversible Mealy
automaton with alphabet \(\Sigma\) generates a finite group, in time
\({\cal O}(\Sigma\log{\Sigma})\).
It is decidable whether it generates a free semigroup, in time
\({\cal O}(\Sigma\log{\Sigma})\).

It is decidable whether a two-letter invertible-reversible Mealy
automaton with stateset \(A\) generates a finite group, in time
\({\cal O}(A\log{A})\).
\end{theorem}

Up to now, the only methods to conclude infiniteness of
automaton groups were to prove the existence
of an element of infinite order~\cite[{\sf
    FindElementOfInfiniteOrder}]{sav}\cite[{\sf SIZE\_FR}]{FR}, using
Sidki's fundamental work~\cite{sidkiconjugacy, sidki}, or to
test level transitivity~\cite[{\sf IsLevelTransitive}]{FR}. All these
methods give sufficient but not necessary conditions.

To illustrate the actual efficiency of the \(\mz\dz\)-triviality as an
algorithm to test finiteness, let us consider the 2-letter 6-state
invertible-reversible Mealy automata. Bireversible Mealy automata are
particular invertible-reversible Mealy automata and 
an invertible-reversible automaton generates a finite group only if it
is bireversible~\cite{AKLMP12}. 
 Testing the
\(\mz\dz\)-triviality of the 3446 bireversible 2-letter 6-states Mealy
automata takes 751ms\footnote{Timings obtained 
  on an Intel Xeon computer with clock   speed 2.13GHz; programs
  written in \GAP~\cite{GAP4}.}, while applying {\sf
  FindElementOfInfiniteOrder}, {\sf SIZE\_FR} or {\sf
  IsLevelTransitive} to determine the infinity of the group
generated by the particular bireversible 2-letter 6-state Mealy
automaton of Figure~\ref{fig-portrait}(a) is unsuccessful after
three weeks of computation.

\subparagraph*{Acknowledgments} I would like to thank Jean Mairesse and
Matthieu Picantin for numerous discussions around this topic.

\bibliographystyle{plain}
\bibliography{./stacs010.klimann}

\end{document}